\documentclass{article}
\usepackage[utf8]{inputenc}

\title{Breaking the Bellman-Ford Shortest-Path Bound
\footnote{This paper is based upon work supported by Science, Technology, \& Innovation Funding Authority (STDF) under grant number 45542}}
\author{Amr Elmasry \\  Egypt University of Informatics and Alexandria University \\ amr.elmasry@eui.edu.eg and elmasry@alexu.edu.eg}

\usepackage[square,sort,comma,numbers]{natbib}
\usepackage{algorithm}
\usepackage{caption}
\usepackage{algpseudocode}
\usepackage{array,colortbl,xcolor}
\usepackage{placeins}
\usepackage{amsfonts}
\usepackage{amssymb}
\usepackage{amssymb,amsmath,amsthm, amsfonts}
\usepackage{csquotes}
\MakeOuterQuote{"}

\usepackage[colorlinks]{hyperref}
\usepackage{graphicx}
\usepackage{epstopdf}

\floatname{algorithm}{Implementation}

\newtheorem{theorem}{Theorem}
\newtheorem{lemma}{Lemma}
\newtheorem{corollary}{Corollary}
\begin{document}

\date{}
\maketitle
In this paper we give a single-source shortest-path algorithm that breaks, after over 65 years, the $O(n \cdot m)$ bound for the running time of the Bellman-Ford-Moore algorithm, where $n$ is the number of vertices and $m$ is the number of arcs of the graph. Our algorithm converts the input graph to a graph with nonnegative weights by performing at most $\min(2 \cdot \sqrt{n},2 \cdot \sqrt{m/\log n})$ calls to a modified version of Dijkstra's algorithm, such that the shortest-path trees are the same for the new graph as those for the original. 
When Dijkstra's algorithm is implemented using Fibonacci heaps, the running time of our algorithm is therefore $O(\sqrt{n} \cdot m + n \cdot \sqrt{m \log n})$.

\paragraph{keywords:} Combinatorial algorithms, Shortest paths, Labeling methods, Dijkstra's algorithm, Negative arcs, Negative cycles
%-----------------------------------------------------------------------------------------------------

\section{Introduction}

The shortest-path problem is one of the classical topics in graph algorithms. 
Consider a directed weighted graph $G= (V,A)$, $V$ is the list of $n$ vertices, $A$ is the list of $m$ arcs,  
and $\ell : A \to \mathbb{R}$ is a length function, where $\ell(u,v)$ is the weight of the arc $(u,v)$.
A shortest path between two vertices is a path of arcs with the minimum total weight. 
The single-source version is to find the shortest paths from a given source $s$ to all other vertices.
The shortest path is undefined if $G$ has a cycle with negative total weight accessible from the source. 
The objective is to get a shortest-path tree from $s$ to all the vertices in $G$, or to alert that $G$ has a negative cycle. 

Since Bellman \citep{bellman1958routing}, Ford \citep{ford1962flows}, and Moore \citep{moore1959shortest} have developed their $O(n \cdot m)$ shortest-path algorithm, several attempts were unsuccessful to break this asymptotic worst-case bound, except for some special cases.
Nevertheless, several heuristics were developed to outperform the Bellman-Ford-Moore algorithm in practice, still not surpassing the theoretical bound. These algorithms include \citep{ElmasryS19,pape1974implementation,pallottino1984shortest,goldberg1993heuristic,Tarjan81}.
If the graph contains cycles of negative weight, all the aforementioned algorithms would report it but most likely not as soon as possible. In the literature there are several algorithms with the primary objective of promptly detecting if a negative cycle exists \citep{cherkassky2009shortest,cherkassky1999negative,goldfarb1991shortest,Tarjan81,wong2005negative}. 
The asymptotic worst-case running times of all these algorithms are at least as the Bellman-Ford bound.   
Other algorithms that allow negative arc weights, but whose running times depend on the weights of the arcs, rely on a technique called {\it scaling} \cite{cohen2017negative,gabow1989faster,goldberg1995scaling}. 
Most notable is Goldberg's algorithm that runs in $O(\sqrt{n} \cdot m \log \mathcal{N})$ time, where $\mathcal{N}$ is the absolute value of the largest negative arc weight \citep{goldberg1995scaling}. 
The well-known Dijkstra algorithm only works efficiently for graphs with nonnegative arc weights \citep{dijkstra1959note}.
Dijkstra's algorithm when implemented using Fibonacci heaps runs in $O(m + n \log n)$ time \citep{fredman1987fibonacci}. 
If there are arcs with negative weights, Dijkstra's algorithm can be used but without a polynomial-time guarantee.
Some suggestions \citep{Dinitz2010,nakayama2013} are to handle the problem by repeatedly applying Dijkstra's algorithm to a series of subgraphs of the original. Such algorithms are referred to as Dijkstra-based algorithms. The running times of the Dijkstra-based algorithms depend on the number of negative arcs and their distribution within the graph. In particular, the algorithm in \citep{Dinitz2010} calls Dijkstra's algorithm at most $k+2$ times, where $k$ is the maximum number of negative arcs on a path of the shortest-path tree.
When the number of negative arcs is large, the performance of these Dijkstra-based algorithms degrades and they cannot compete with the other aforementioned algorithms in practice. 

On the other hand, more efficient shortest-path algorithms with better theoretical bounds exist for some special graph families like: planar graphs \citep{klein2010} and Euclidean graphs \citep{sedgewick1986}. 
If the arc weights are nonnegative integers, special priority queues were developed to solve the shortest-path problem more efficiently, depending on the weights of the arcs \citep{ahuja1990faster, thorup2000}.
Most notable is that the single-source shortest-path problem can be solved in $O(n+m)$ time for acyclic graphs using topological sorting \citep{cormen1998introduction}. 

Recently, Fineman \cite{Fine24} introduced a randomized single-source shortest-path algorithm that runs in expected {{\~{O}}$(m \cdot n^{8/9})$ time. Our algorithm is deterministic and asymptotically faster.

In this paper we introduce a new algorithm for the single-source shortest-path problem breaking the $O(n \cdot m)$ time bound that endured over 65 years. Our algorithm repeatedly calls: i) a linear-time shortest-path algorithm for an acyclic graph with nonpositive weights and ii) Dijkstra's algorithm on a graph with nonnegative weights. Afterwards, the original graph is transformed to a graph with nonnegative weights that has the same shortest-path trees as the original graph. One more call to Dijkstra's algorithm on this transformed graph would then finish the job. 
Our algorithm runs in $O(\sqrt{n} \cdot m + n \cdot \sqrt{m \log n})$ time. 

We have implemented our algorithm to anticipate its performance in practice, and ran it on a large set of test cases. Experimental results indicate that our algorithm is practically efficient, and much faster than all the competitors. We shall report these results later in a separate study.

%-----------------------------------------------------------------------------------------------------

\section{Background}

\subsection{Label correcting methods}
Most shortest-path algorithms are based on the general \textit{label-correcting} method \citep{Bertsekas1993,cormen1998introduction,gallo1988shortest,shier1981properties}. For every vertex $u$, the method maintains a potential value $d[u]$ (tentative distance from the source $s$) and a parent pointer $p[u]$, and update them throughout the algorithm. %The parent pointers form a parent graph $\mathcal{T}_p$. 
The method starts by setting $d[s]=0$;
and for every other vertex: $d[u]=\infty$, $p[u]=nil$. 

The \textit{scan} operation, defined on a vertex $u$, checks all outgoing arcs from $u$ for \textit{relaxation}.
For every vertex $v$ in the adjacency list of $u$, $u.adj[~]$, the arc $(u,v)$ is relaxed if $d[v] > \ell(u,v) + d[u]$ by setting $d[v] \leftarrow \ell(u,v) + d[u]$ and setting $p[v] \leftarrow u$.
The method works in rounds until no more arcs can be relaxed. For each round, the scan operation is applied to some and possibly all the vertices. Different strategies for selecting which vertices to be scanned and their scanning order lead to different algorithms.
When the algorithm terminates, if $G$ has no negative cycles, the final value $d[u]$ for every vertex $u$ is the shortest-path distance from $s$ to $u$, and the parent pointers constitute a shortest-path tree.
If $G$ has negative cycles, the label-correcting method can be easily modified to find such a cycle and terminate.

\subsection{The Bellman-Ford-Moore algorithm}
The Bellman-Ford-Moore algorithm scans, in each round, all the vertices of the graph in arbitrary order. Other than the source vertex, the number of vertices with correctly settled potentials is at least $i$ after the $i$-th round. In accordance, if there are no negative cycles, the shortest paths are correctly computed after at most $n-1$ rounds. The running time of the algorithm is $O(n \cdot m)$.

\algblockdefx{MRepeat}{EndRepeat}{\textbf{repeat}}{}
\algnotext{EndRepeat}

\begin{algorithm}[ht]
\caption{The Bellman-Ford-Moore Algorithm}
\begin{algorithmic}
\For{\texttt{all $u \in V$}}
        \State $d[u] \leftarrow \infty$
				\State $p[u] \leftarrow nil$
\EndFor			
\State $d[s] \leftarrow 0$
\MRepeat{ $n-1$ times}
	\For{\texttt{all $(u,v) \in A$ }}
            \If{$d[v] > \ell(u,v) + d[u]$}
							\State $d[v] \leftarrow \ell(u,v) + d[u]$
							\State $p[v] \leftarrow u$
            \EndIf
  \EndFor
\EndRepeat
\end{algorithmic}
\end{algorithm}

\subsection{Dijkstra's algorithm}
Dijkstra's algorithm \citep{dijkstra1959note} is a label-correcting method that works efficiently for graphs with nonnegative arc weights. 
Each round, the algorithm selects a vertex with the minimum potential value among the unscanned vertices to be scanned next. For graphs with nonnegative arc weights, once a vertex is scanned it need not be scanned again. The worst-case complexity of Dijkstra's algorithm in this case depends on the data structure used for finding the vertex with the minimum potential value: when implemented using arrays or linked lists, the algorithm runs in $O(n^2)$ time; when implemented using binary heaps, the algorithm runs in $O(m \lg n)$ time; when implemented using Fibonacci heaps, the algorithm runs in $O(m + n \lg n)$ time \citep{fredman1987fibonacci}.
We use the subroutine {\it Create-Priority-Queue}$(V,d)$ to initialize a priority queue (heap) $Q$ with the items in the list of vertices $V$ using the corresponding values in a list of potentials $d$. The subroutine {\it Extract-Min}$(Q)$ returns the item with the minimum value in $Q$. The subroutine {\it Decrease}$(Q,v,d[v])$ decreases the value corresponding to the vertex $v$ in $Q$ to its new potential value $d[v]$.

\algblockdefx{MRepeat}{EndRepeat}{\textbf{repeat}}{}
\algnotext{EndRepeat}

\begin{algorithm}[ht]
\caption{Dijkstra's Algorithm}
\begin{algorithmic}
\For{\texttt{all $u \in V$}}
        \State $d[u] \leftarrow \infty$
				\State $p[u] \leftarrow nil$
\EndFor			
\State $d[s] \leftarrow 0$
\State $Q$ $\gets$ Create-Priority-Queue($V, d$)
\While{\textit{$Q$ is not empty}}
        \State $u \gets$ Extract-Min($Q$)
				\For{\texttt{all $v \in u.adj[~]$ }}
            \If{$d[v] > \ell(u,v) + d[u]$}
							\State $d[v] \leftarrow \ell(u,v) + d[u]$
							\State $p[v] \leftarrow u$
							\State Decrease($Q,v,d[v]$)
            \EndIf
         \EndFor
\EndWhile
\end{algorithmic}
\end{algorithm}

\subsection{Shortest paths for acyclic graphs}
Topological sorting of an acyclic graph is to produce a linear ordering of the vertcies of the graph where $u$ appears before $v$ if $(u,v)$ is an arc in the graph. This can be done by keeping track of a list $L$ of the vertices having no ingoing arcs. At every iteration, an arbitrary vertex is removed from $L$ and appended to the output, and subsequently  all the arcs outgoing from this vertex are removed from the graph and $L$ is updated in accordance \citep{cormen1998introduction}. As the resulting graph is always acyclic, there is at least one vertex in $L$ until all the vertices are ordered. We use the subroutine {\it Vertices-Topological-Sort}$(G)$ to produce such an ordering from an acyclic graph $G$. 

The shortest paths in an acyclic graph with $n$ vertices and $m$ arcs, even with negative arcs, can then be found in $O(n+m)$ time using topological sorting \citep{cormen1998introduction}. Once the topological ordering is available, the vertices are scanned in order to find the shortest path distances. 

\begin{algorithm}
\caption{Acyclic-Shortest-Paths}
\begin{algorithmic}
\For{\texttt{all $u \in V$}}
        \State $d[u] \leftarrow \infty$
				\State $p[u] \leftarrow nil$
\EndFor			
\State $d[s] \leftarrow 0$
    \State $Q$ $\gets$ \textit{Vertices-Topological-Sort$(G)$}
    \While {$Q$ is not empty}
        \State  $u \gets$ \textit{dequeue$(Q)$}
        \For{\texttt{all $v \in u.adj[~]$ }}
        \If{$ d[v] > \ell(u,v) + d[u]$}
            \State $d[v] \gets \ell(u,v) + d[u]$
						\State $p[v] \gets u$
         \EndIf   
        \EndFor
    \EndWhile
		\State return($G$)
\end{algorithmic}
\end{algorithm}  

\subsection{Potentials and reduced weights}
Given a list of potential values $d$ corresponding to the list of vertices $V$, the reduced weight $l_d: A \to \mathbb{R}$ for an arc $(u,v) \in A$ is defined
\[\ell_d(u,v) = \ell(u,v) + d[u] - d[v].\]
Let $G_d$ be the graph $G$ with the reduced weights instead of the originals, and call it the {\it reduced-cost graph}.
For any path $P = \left< u_0,u_1,\dots,u_k \right>$ in $G$, let $\ell(P) = \sum_{i=1}^k \ell(u_{i-1},u_i)$ and $\ell_d(P) = \sum_{i=1}^k \ell_d(u_{i-1},u_i)$.
For a path $P$, it directly follows that $\ell_d(P) = \ell(P) + d[u_0] - d[u_k]$. 
Hence, the shortest path between any two vertices is the same for the graph $G$ and the reduced-cost graph $G_d$. 
For a cycle $C$, $\ell_d(C) = \ell(C)$.
Hence, the graph $G$ has a negative cycle if and only if $G_d$ has a negative cycle.
Our objective is to find the potentials that make the reduced weights for $G_d$ all nonnegative.
Once the reduced-cost graph has no negative weights, we can apply Dijkstra's algorithm to get the shortest-path tree from the given source. This tree is also the shortest-path tree for the original graph.

\subsection{The admissible subgraph and zero cycles}
A subgraph of the original graph $G$ that only includes the arcs with the nonpositive reduced weights is called the {\it admissible subgraph} $G^-$.
Typically, if there is a cycle in the admissible subgraph, then there is a negative cycle in the original graph.
The only exception is for a cycle with zero total weight, i.e., all its arcs have zero weight. 
In this case, since all the vertices on a zero cycle have the same shortest-path values, we aim to get rid of those cycles. As in \citep{goldberg1995scaling}, we can contract the vertices of the zero cycles after identifying those cycles in linear time by finding the strongly connected components of $G^-$ \citep{Tarjan72}. 
If there is a negative arc that connects two vertices of a strongly connected component of $G^-$, a negative cycle is found. We would therefore assume in the sequel that the zero cycles are first removed by our algorithm, and hence the admissible subgraph $G^-$ is acyclic as long as $G$ has no negative cycle.
Once we come up with potentials for the contracted graph resulting in nonnegative reduced weights for all the arcs, the reduction can be directly extended to the original graph.

\section{Preliminaries}
 
\subsection{Snakes}
In a graph with negative arc weights, we define a {\it snake} as a maximal-length segment that has zero or more consecutive arcs with zero weight followed by one arc with negative weight.
The zero-weight arcs form the {\it tail} of the snake, and the negative arc is its {\it head}. 
Obviously, if the target vertex of the head of a snake touches its tail, it forms a negative cycle.
The length (size) of a snake is the number of arcs it contains.
A snake is uniquely identified by its head, and may share parts of its tail with other snakes.
We call a segment with nonnegative arcs, that are not on a snake's tail, a {\it passive segment}.
In general, any path is composed of consecutive snakes possibly alternating with passive segments.
%A graph that has no negative arcs has no snakes.
%We call a path with nonpositive arcs, where at least one of them is negative, a {\it negative segment}. Note that a snake is a special negative segment.  
%When the head of a snake touches the tail of another snake, the first snake together with the front part of the second form a negative segment.
%A {\it minimal passive segment} is a path with consecutive zero-weight arcs that ends with one positive arc, if any.
%We also call a path with consecutive positive arcs a {\it positive segment}, and that with consecutive negative arcs a {\it negative segment}.

Throughout our algorithm, as a result of applying different potential values, the reduced weights on the arcs change, and accordingly the snakes in the reduced-cost graph would be dynamically changing their positions, lengths, and may as well disappear. 
Initially, every negative arc is a head of a snake.
We call the negative arc in the input graph where a snake originates at the beginning of the algorithm, the {\it origin} of the snake. 
If the reduced weight of the head of a snake becomes zero and that of a succeeding arc becomes negative, while those of the other arcs of its tail are still zeros, we say that the snake {\it expands}. 
A snake may expand on several succeeding arcs simultaneously forming several snakes.
If the reduced weight of the head of a snake becomes zero as well as the arcs of a succeeding passive segment up to the tail of another snake, we say that the first snake {\it connects} to the latter snake forming one larger snake.
Alternatively, if the reduced weight of the head of a snake becomes zero and the succeeding arc on a path is nonnegative, the snake becomes a passive segment on that path. 
If the snake becomes a passive segment on all the outgoing paths, we say that the snake {\it vanishes}.
With every iteration of our algorithm, every snake in the reduced-cost graph may: expand, connect to another snake, or vanish. 
%
%
%When a snake connects to another it cuts the latter and inherits its front part including its head. The head is called the {\it inherited head} and the origin of the latter is a {\it contributing origin} to this inherited head.
%The path from a contributing origin to the inherited head contains the front part of the resulting snake that is called a {\it residual}.
%
%Once all the snakes that stem from an origin convert to passive segments, this origin becomes a {\it passive origin}.
%For any snake, we say that the snake is a {\it native snake} to the path that contains the whole snake's body from the origin to the head. 
%Alternatively, if the head of a snake lies on a path but not its origin or any contributing origin, the snake is said to be an {\it intruding snake} to this path. 

\subsection{Insights}

The goal of our algorithm is to come up with a list of potential values such that the resulting reduced-cost graph would have no negative arcs, or otherwise the graph is shown to have a negative cycle.  
The potential values are initialized to zeros.
Every iteration, some of the potentials are changed implying changes in the reduced-cost graph.
In accordance, the negative weights are gradually pushed forward through different paths of the reduced-cost graph until they are mended, or until a negative cycle is discovered if exists. 
The good news is that all the negative weights disappear in at most $\min(2 \cdot\sqrt{n},2 \cdot \sqrt{m/\log n})$ iterations.

We shall use shortest-path computations to perform this task of pushing the negative weights forward, and hence eliminating them allover the graph.
It is possible to find the shortest paths within the admissible subgraph in linear time as it is acyclic (assuming that the graph has no negative cycles).  It is also possible to efficiently find shortest paths for a graph with nonnegative weights using Dijkstra's algorithm. We apply the two procedures alternately and repeatedly.
Before the end of every iteration, since the negative weights are obstacles for the upcoming calls of Dijkstra's algorithm, we need to gradually mend these negative weights. In accordance, we use the potential values to replace the current weights with the reduced weights, and subsequently reset the potentials to zeros.
Evidently, the weights change with every iteration. 
%We call the weights at the current iteration of the algorithm, the {\it current weights}.

Consider the graph $G'$ formed by adding an artificial source vertex $s'$ with zero-weight arcs from $s'$ to the source vertices of the origins of the snakes in $G$. Let $\mathcal{T}$ be the shortest-path tree in $G'$ from $s'$. Let $\delta[x]$ be the value of the shortest path in $\mathcal{T}$ from $s'$ to vertex $x$.
The graph $G$ can be converted to a graph with nonnegative weights having the same shortest-path trees when setting the potential value for every vertex $x$ as $d[x] \gets \min\{\delta[x],0\}$ and replacing the actual weights by the nonnegative reduced weights, as indicated by the next lemma.

\begin{lemma}
\label{potentials}
By setting the potential value of every vertex $x$ to $d[x] \gets \min\{\delta[x],0\}$, the reduced weights of all the arcs are nonnegative.
\end{lemma}

\begin{proof}
For any arc $(u,v)$ in the graph, the triangle-inequality property for the shortest paths indicates that $\delta[v] \leq \ell(u,v) + \delta[u]$. In all the following cases $\ell_d(u,v)$ turns out to be nonnegative.
\begin{enumerate}
\addtolength{\itemindent}{1cm}
\item[case 1.] $\delta[u] \leq 0$ and $\delta[v] \leq 0$: we set $d[u] \gets \delta[u]$ and $d[v] \gets \delta[v]$. Hence, $\ell_d(u,v) = \ell(u,v) + d[u] - d[v] = \ell(u,v) + \delta[u] - \delta[v] \geq 0$.
\item[case 2.] $\delta[u] \leq 0$ and $\delta[v]>0$: we set $d[u] \gets \delta[u]$ and $d[v] \gets 0$. Hence, $\ell_d(u,v) = \ell(u,v) + d[u] - d[v] = \ell(u,v) + \delta[u] \geq \delta[v]> 0$.
\item[case 3.] $\delta[u] > 0$: It must be the case that $\ell(u,v) \geq 0$, for otherwise an arc $(s',u)$ would have been added in $G'$ with $\ell(s',u)=0$ resulting in $\delta[u] \leq 0$. 
\vspace{-.1in}
\begin{enumerate}
\addtolength{\itemindent}{1.3cm}
 \item[case 3.1.] $\delta[v] \geq 0$: we set $d[u] \gets 0$ and $d[v] \gets 0$. Hence, $\ell_d(u,v) = \ell(u,v) + d[u] - d[v] = \ell(u,v) \geq 0$.
 \item[case 3.2.] $\delta[v] < 0$: we set $d[u] \gets 0$ and $d[v] \gets \delta[v]$. Hence, $\ell_d(u,v) = \ell(u,v) + d[u] - d[v] = \ell(u,v) - \delta[v] > 0$.
\end{enumerate}  
\vspace{-.1in}
\end{enumerate}  
\vspace{-.15in}
\end{proof}

While our procedures operate directly on the weights of the graph, the proof follows by tracing the snakes in the reduced-cost graph.
The intuition behind the proof is as follows. 
If there are no negative cycles, the negative weights are pushed forward on the paths of the shortest-path tree $\mathcal{T}$ until mended, over paths that obviously have at most $n-1$ arcs each.
In addition, the nonnegative arcs will remain nonnegative while negative arcs gradually become nonnegative throughout the algorithm.
After applying the first procedure (shortest paths for the admissible subgraph that is an acyclic graph), the negative reduced weights become nonnegative, while possible new negative reduced weights appear for the arcs following the paths of the admissible subgraph. In accordance, every snake of the reduced-cost graph either expands in size or converts to a passive segment. As a result, the negative arcs are either mended or pushed forward. 
After applying the second procedure (Dijkstra's algorithm on a subgraph with nonnegative weights),
at least the first alive snake with respect to the reduced-cost graph on every path of the shortest-path tree $\mathcal{T}$ either connects to the origin of the next snake forming a larger snake or converts to a passive segment. As a result, the heads of all the first snakes become farther from $s'$.

\section{The main procedures}

Our algorithm includes three subroutines that are executed consecutively and repeatedly until the reduced-cost graph has no negative arcs: the expansion procedure, the connection procedure, and the adjust-weights procedure. Every iteration, the potential values on all vertices are initially zeros and keep decreasing throughout the iteration.

\subsection{The expansion procedure}

The first procedure, which we call the {\it expansion procedure}, operates on the admissible subgraph $G^-$. When $G$ has no negative cycles, the admissible subgraph $G^-$ is acyclic. The expansion procedure finds the shortest-path values from the source vertices to every vertex within $G^-$. 
%Let $\mathcal{T}^-$ be the resulting shortest-path forest. 

To find the shortest-path values in an acyclic graph in linear time, the vertices of the graph are topologically sorted by the subroutine \textit{Vertices-Topological-Sort$(G^-)$} forming a list of vertices $Q$. As a precondition to the expansion procedure, the potential values on all the vertices are set to zeros. The list $Q$ is scanned in order, and for each vertex $u$ in $Q$ and every vertex $v$ in the adjacency list of $u$, $u.adj[~]$, the arc $(u,v)$ is relaxed if $d[v] > \ell(u,v) + d[u]$ by setting $d[v] \leftarrow \ell(u,v) + d[u]$.

\begin{algorithm}
\begin{algorithmic}
\Procedure{Expand}{$G$}
    \State $Q$ $\gets$ \textit{Vertices-Topological-Sort$(G^-)$}
    \While {$Q$ is not empty}
        \State  $u \gets$ \textit{dequeue$(Q)$}
        \For{\texttt{all $v \in u.adj[~]$ }}
        \If{$ d[v] > \ell(u,v) + d[u]$ }
            \State $d[v] \gets \ell(u,v) + d[u]$
         \EndIf   
        \EndFor
    \EndWhile
		\State return($G$)
\EndProcedure
\end{algorithmic}
\end{algorithm}

The rough purpose of this procedure is to simultaneously push all the negative weights at least one arc forward, and hence expand the snakes of the resulting reduced-cost graph compared to the original. More precisely, all the negative arcs of the admissible subgraph become nonnegative while their positive immediate-successor arcs, if any, may become negative in return. In conclusion, every snake either expands or vanishes.

\begin{lemma}
\label{extend-shortest}
Consider any snake when applying the expansion procedure. On any specified path in the resulting reduced-cost graph, the snake either expands in size or converts to a passive segment.
\end{lemma}

\begin{proof}

Consider the reduced weight of any arc $(u, v)$ in $G$ after applying the expansion procedure. 
We have $\ell_{d}(u,v) = \ell(u,v) + d[u] - d[v]$.

If $(u, v)$ was in $G^-$, the expansion procedure guarantees the shortest-path property on the arcs of $G^-$ implying that $d[v] \leq \ell(u,v) + d[u]$, and so $\ell_{d}(u,v) \geq 0$. Hence, the reduced weights of all the arcs that were in $G^-$ become nonnegative.

If $(u, v)$ was not in $G^-$, we have $\ell(u,v) > 0$. 
Consider the case where $d[u] = 0$ after the expansion procedure, indicating that there is no path in $G^-$ with a snake on it ending at $u$. As all the potential values are nonpositive by Lemma \ref{potentials}, then $d[v] \leq 0$. Hence, $\ell_{d}(u,v) \geq \ell(u,v) >0$. The reduced weights of the positive arcs that do not follow a path of $G^-$ thus remain positive.

The only case where $\ell_{d}(u,v)$ possibly turns negative after the expansion procedure is if $(u,v)$ was not in $G^-$ and $d[u] < 0$. 
In such a case, there is path in $G^-$ ending up at $u$ with a snake on it. As the reduced weights of all the arcs on this path become zeros, the snake turns into a larger snake whose head is $(u,v)$. 
On the other hand, if $\ell_{d}(u,v)$ remains positive, this means that any path with a snake entering $u$ becomes a passive segment.

It is possible that a snake loses a lower part of its tail. Assume that $(u,v)$ was on the tail of a snake before the procedure, and another ingoing snake connects with the first at vertex $v$. As a result, the value of $d[v]$ may decrease and $\ell_d(u,v)$ may become positive in accordance. In such a case, this ingoing second snake connects to the upper part of the first, cutting the lower part of its tail being a passive segment. As a result, the first snake vanishes and the second expands. 
\end{proof}

\subsection{The connection procedure}

Complementing the expansion procedure, the resulting potential values are used in another shortest-path computation, but now for a graph with nonnegative weights.

In more details, the connection procedure operates on a graph $G^+$, resulting from $G'$ by ignoring the negative arcs (assume that their weights are infinity). Since all the weights of $G^+$ are nonnegative, we apply a modified version of Dijkstra's algorithm having the priority queue $Q$ initialized with the potential values that are the shortest-path values resulting from the expansion procedure. As is customary with the original algorithm, every iteration Dijkstra's algorithm selects the vertex with the minimum potential to be scanned next, using the subroutine {\it Extract-Min}$(Q)$, and if possible decreases the potentials of adjacent vertices in accordance. Once a vertex is scanned it is not scanned again.

\begin{algorithm}
\begin{algorithmic}
\Procedure{Connect}{$G$}
    \State Q $\gets$ Create-Priority-Queue($V, d$)
    \While{\textit{Q is not empty}}
        \State $u \gets$ Extract-Min($Q$)
        \For{\texttt{all $v \in u.adj[~]$ }}
            \If{$\ell(u,v) \geq 0$ and $d[v] > d[u] + \ell(u,v)$}
							\State $d[v] \leftarrow \ell(u,v) + d[u]$
							\State Decrease($Q,v,d[v]$)
						\EndIf
        \EndFor
        \EndWhile
			\State return ($G$)	
\EndProcedure
\end{algorithmic}
\end{algorithm}

The rough purpose of this procedure is to mend the first arc with negative reduced weight on every path of the shortest-path tree $\mathcal{T}$.
More precisely, the reduced weight of the head of the first snake on every path of the shortest path tree either becomes zero and connects to the head of the next snake in one larger snake, or otherwise becomes nonnegative and the snake vanishes. 

\begin{lemma}
\label{connect-shortest}
Consider the first snake on each path of the shortest-path tree $\mathcal{T}$ when applying the connection procedure. In the resulting reduced-cost graph, the snake either connects to the next snake or converts to a passive segment.
\end{lemma}

\begin{proof}

Consider a path $P$ on the shortest-path tree $\mathcal{T}$. The potential values on all the vertices of the first snake on $P$ are the shortest-path values within this call to the expansion procedure. These values will not change by the connection procedure, implying that the reduced weight of any arc $(u, v)$ of these arcs is $\ell_{d}(u,v) = \ell(u,v) + d[u] - d[v] = \ell(u,v) + \delta[u] - \delta[v] = 0$.

Let $x$ be the target vertex of the head of the first snake on $P$. We show by induction on the vertices of $P$ from $x$ up to the first vertex with positive shortest-path value or up to the source vertex of the head of the next snake on $P$, whichever closer, that the potential on each of these vertices will attain its shortest-path value before the vertex is scanned within this call to the connection procedure.
Following the expansion procedure, the potential on vertex $x$ attains the value of the shortest path to $x$, i.e. $d[x] = \delta[x]$, initially fulfilling the induction hypothesis.
Consider any arc $(u, v)$ on $P$ where $\ell(u,v) \geq 0$. From the shortest-path property, we have $\delta[v] = \ell(u,v) + \delta[u]$ implying that $\delta[v] \geq \delta[u] $. 

For any vertex $v$ with $\delta[v] \leq 0$, we have $d[v] \geq \delta[v]$ throughout the algorithm. The case where $d[v] = \delta[v]$ before $v$ is scanned directly fulfills the induction hypothesis. 
Consider the case where $d[v] > \delta[v]$, indicating that $d[v] > \delta[u]$.
Assume that the induction hypothesis is fulfilled up to vertex $u$. By the induction hypothesis, vertex $u$ will be scanned only when $d[u] = \delta[u]$, which is shown to be less than $d[v]$. As the connection procedure chooses the vertex with the minimum potential value to be scanned in each iteration, vertex $u$ will be scanned before vertex $v$. Once vertex $u$ is scanned, the potential on $v$ is settled to $d[v] = \ell(u,v) + d[u] = \ell(u,v) + \delta[u] = \delta[v]$, also fulfilling the induction hypothesis.   
In accordance, the potential values on all the vertices up to the source vertex of the head of the second snake on $P$ will reach their shortest-path values within this call to the connection procedure, as long as these shortest-path values are nonpositive. Consequently, for any of the arcs on this segment, $\ell_{d}(u,v)=\ell(u,v) + d[u] - d[v] = \ell(u,v) + \delta[u] - \delta[v] = 0$. 
It follows that the first two snakes on $P$ connect to form one larger snake in the reduced-cost graph.

For the case where $v$ is the first vertex on $P$ with $\delta[v]>0$, the connection procedure ends up with $d[u]=\delta[u]$ and $d[v]=0$ (by Lemma \ref{potentials}). Since $\delta[v]>0$ in this case, the shortest-path property implies $\ell(u,v) + \delta[u] > 0$. Consequently, $\ell_{d}(u,v) = \ell(u,v) + d[u] - d[v]= \ell(u,v) + \delta[u] > 0$. It follows that the first snake on $P$ converts into a passive segment in the reduced-cost graph.
\end{proof}

\begin{lemma}
\label{remains-nonnegative}
The arcs with nonnegative weights before the connection procedure attain nonnegative reduced weights after the procedure.
\end{lemma}

\begin{proof}
Consider an arc $(u,v)$ where $\ell(u,v) \geq 0$. Every vertex is scanned exactly once within the connection procedure. When $u$ is scanned, $(u,v)$ is relaxed and $d[v] \leq \ell(u,v) + d[u]$ implying $\ell_d(u,v) \geq 0$. Since all potentials never increase within the connection procedure, the only way for $\ell_d(u,v)$ to presumably end up negative is that $d[u]$ would decrease after $u$ is scanned. For this to happen, there must exist a vertex $q$ with $\ell(q,u) \geq 0$ such that $q$ is scanned after $u$. 
The connection procedure scans the vertex with the minimum potential in every iteration and only deals with nonnegative arcs. It follows that the potentials of the scanned vertices form a non-decreasing sequence. In particular, $d[u] \leq d[q]$ when $q$ is scanned. But then, relaxing $(q,u)$ would result in $d[u] \leftarrow \ell(q,u) + d[q]$. Since $\ell(q,u) \geq 0$, $d[u]$ would never decrease after $u$ is scanned. This contradiction ensures that $\ell_d(u,v) \geq 0$ after the procedure. 
\end{proof}

\subsection{Adjusting weights}

In this procedure, the current weights of the graph $G$ are replaced with the reduced weights using the list of potential values $d$. All the potential values are then reset to zeros. 
%
%More precisely, for every arc $(u,v)$, we set $\ell(u,v) \leftarrow \ell(u,v) + d(u) - d(v)$. 
In accordance, the reduced weights remain the same. 

\begin{algorithm}
\begin{algorithmic}
\Procedure{Adjust-Weights}{$G$}
    \For{\texttt{all $(u,v) \in E$}}
        \State $\ell(u,v) \leftarrow \ell(u,v) + d[u] - d[v]$
    \EndFor
		\For{\texttt{all $u \in V$}}
        \State $d[u] \leftarrow 0$
    \EndFor
		\State return($G$)
\EndProcedure
\end{algorithmic}
\end{algorithm}

\begin{lemma}
\label{change-weights}
When applying the adjust-weights procedure, after the expansion and the connection procedures, 
all the arcs with nonnegative weights before the two procedures remain nonnegative.
%In addition, the weights of some of the origins will become nonnegative. 
\end{lemma}

\begin{proof}
By Lemma \ref{remains-nonnegative}, when applying the expansion followed by the connection procedures, the reduced weights of the nonnegative arcs remain nonnegative. Hence, the new weights of these arcs become nonnegative when using the reduced weights to replace the current weights.
%The reduced weights of the origins that are not heads of snakes are nonnegative after the connection procedure, and so the new weights of these origins become nonnegative after the adjust-weights procedure.  
\end{proof}

\section{The snakes algorithm}

The expansion and connection procedures offer great machinery to handle snakes. The former makes each snake longer or coverts it into a passive segment, while the latter connects the first surviving snake on each path with the origin of another snake or coverts it into a passive segment. 

In combination, the two procedures expand each snake by at least one arc and connect the first snake to an origin, or otherwise convert the snake to a passive segment.
Obviously, we shall repeatedly execute an interleaved combination of the two procedures until the reduced-cost graph has no negative weights. We prove that the combination will be executed at most $\sqrt{2n}$ times.

\algblockdefx{MRepeat}{EndRepeat}{\textbf{repeat}}{}
\algnotext{EndRepeat}

\begin{algorithm}[ht]
\caption{The Snakes Algorithm}
\begin{algorithmic}[1]
\For{\texttt{all $u \in V$}}
        \State $d[u] \leftarrow 0$
\EndFor			
\While{$G$ still has negative arcs}
\State $G \gets \Call{Expand}{G}$ 
\State $G \gets \Call{Connect}{G}$ 
\State $G \gets \Call{Adjust-Weights}{G}$ 
\EndWhile
\end{algorithmic}
\end{algorithm}

\begin{lemma}
\label{j-times}
After  $j-1$ iterations of the algorithm, the length of every snake will be at least $j$.
\end{lemma}

\begin{proof}
Initially, every negative arc is the head of a snake with length at least 1. 
By Lemma \ref{extend-shortest}, in the reduced-cost graph, the length of every alive snake increases with every iteration. 

Consider an arc $(u,v)$ where $\ell(u,v) < 0$ just before executing the expansion procedure. As a result of the expansion procedure, by Lemma \ref{extend-shortest}, we get $\ell_d(u,v) \geq 0$. The only way for the arc to attain a negative weight at the end of the iteration is that $d[u]$ decreases by the connection procedure. For this to happen, there must exist another snake that would be extended by the connection procedure such that its body reachs $u$ resulting in $(u,v)$ being its new head. As a result, this other snake inherits the head of the fist and is now longer.  

In addition, by Lemma \ref{remains-nonnegative}, no nonnegative arcs turn negative by the connection procedure, and hence no new snakes are created throughout the algorithm.
\end{proof}

\begin{theorem}
We can convert a graph with negative weights, but no negative cycle, to a graph that has no negative weights, such that the two graphs have the same shortest-path tree, using less than $\sqrt{2n}$ calls to a linear-time shortest-path computation for an acyclic graph and to a modified version of Dijkstra's algorithm, where $n$ is the number of vertices of the graph.
\end{theorem}

\begin{proof}
We distinguish between two types of snakes in the reduced-cost graph: the snakes that are the first on the paths of the shortest-path tree $\mathcal{T}$ and the other snakes. 
We show by mathematical induction that before iteration $j$ the heads of the first snakes are at a distance at least $j(j-1)/2$ from $s'$ on every path of the shortest-path tree $\mathcal{T}$.
The hypothesis is obviously true for $j=1$. 

Assume the hypothesis holds after iteration $j-1$ and consider the changes at iteration $j$.
Obviously, the origin of the second snake on any path follows the head of the first, i.e. at a distance at least $j(j-1)/2$ from $s'$ by the induction hypothesis.
By Lemma \ref{j-times}, the length of every snake after iteration $j-1$ is at least $j$. 
By Lemma \ref{connect-shortest}, the connection procedure ensures that each of the first snakes on every path of $\mathcal{T}$ either connects to the second snake forming one snake or vanishes. The lemma also indicates that the potentials on the vertices of the first snakes have already reached their shortest-path values, implying that no first snake would ever be connected to another first snake.
% 
%The reduced weight of a nonnegative arc is always nonnegative, as indicated by Lemma \ref{remains-nonnegative}.
%When the weights are changed, the reduced weights replaces the actual weights.
Once a negative arc is mended its actual weight will forever turn nonnegative, as indicated by Lemma \ref{change-weights}. This will permit the snakes not to be obstructed by negative arcs that are already mended.
It follows that the heads of the first snakes after the $j$th iteration must be at distance at least $j(j-1)/2 + j = j(j+1)/2$ from $s'$. The induction hypothesis then holds after iteration $j$. 

Let $j'$ be the number of iterations executed by the algorithm.  As there are at most $n$ levels on any path in the graph, it follows that $j'(j'+1)/2 \leq n$ indicating that $j'$ is less than $\sqrt{2n}$, and the algorithm must have been terminated by then with a graph that has no negative arcs.
\end{proof}

\begin{corollary}
We can convert a graph with negative weights to a graph that has all nonnegative weights, such that the two graphs have the same shortest-path tree starting from any source, using $O(\sqrt{n} \cdot m + n^{3/2} \cdot \log n)$ time when implementing the connection procedure using Fibonacci heaps, where $n$ is the number of vertices and $m$ is the number of arcs of the underlying graph.
\end{corollary}

\section{Improvement}

To improve the worst-case asymptotic running time of the algorithm even further, we execute the expansion procedure 
$c=\lceil \frac{n \log n}{m} \rceil$ times within every iteration of the algorithm. In accordance, the lengths of the snakes grow by at least this amount in every iteration. 
%When using Fibonacci heaps, this would result that the work per iteration is balanced between the expansion and the connection procedures to be $O(m + n \log n)$. 
In effect, the bound on the number of iterations for the algorithm to terminate would be improved to $O(\min(\sqrt{n}, \sqrt{\frac{m}{\log n}}))$.

\algblockdefx{MRepeat}{EndRepeat}{\textbf{repeat}}{}
\algnotext{EndRepeat}

\begin{algorithm}[ht]
\caption{The Improved Snakes Algorithm}
\begin{algorithmic}[1]
\For{\texttt{all $u \in V$}}
        \State $d[u] \leftarrow 0$
\EndFor			
\While{$G$ still has negative arcs}
\MRepeat{$\lceil \frac{n \log n}{m} \rceil$} times
	\State $G \gets \Call{Expand}{G}$ 
\EndRepeat	
\State $G \gets \Call{Connect}{G}$ 
\State $G \gets \Call{Adjust-Weights}{G}$ 
\EndWhile
\end{algorithmic}
\end{algorithm}

\begin{theorem}
We can convert a graph with negative weights to a graph that has all nonnegative weights, such that the two graphs have the same shortest-path tree starting from any source, using $O(\sqrt{n} \cdot m + n \cdot \sqrt{m \log n})$  time when implementing Dijkstra's algorithm using Fibonacci heaps, where $n$ is the number of vertices and $m$ is the number of arcs of the underlying graph.
\end{theorem}

\begin{proof} 
If $n \log n \leq m$, $c=\lceil \frac{n \log n}{m} \rceil=1$ and the algorithm is the same as before indicating that the total running time is $O(\sqrt{n} \cdot m + n^{3/2} \log n) = O(\sqrt{n} \cdot m)$.

The length of every snake is at least $c \cdot j$ after iteration $j-1$. In accordance, the heads of the first snakes after the $j$th iteration must be at distance at least $c \cdot j(j+1)/2$ from $s'$.
It follows that $c \cdot j'(j'+1)/2 \leq n$ indicating that $j' \leq \sqrt{2n/c} = O(\sqrt{m/\log n})$.
The running time consumed by the expansion procedure in every iteration is $O(m \cdot \frac{n \log n}{m}) = O(n \log n)$.
Using Fibonacci heaps, the running time of the connection procedure per iteration is $O(m + n \log n)$. 

If $n \log n > m$, the running time per iteration is $O(n \log n)$. The total running time of the algorithm in this case is therefore $O(n \log n \cdot \sqrt{m/ \log n}) = O(n \cdot \sqrt{m \log n})$.
\end{proof}

\bibliographystyle{abbrv}
\bibliography{archive}
\end{document}